\documentclass[conference]{IEEEtran}
\usepackage[utf8]{inputenc}
\usepackage{amsmath}
\usepackage{ragged2e}
\usepackage{algorithm}
\usepackage{amsfonts}
\usepackage{amssymb}
\usepackage{color,soul}
\usepackage{scalerel}
\usepackage[T1]{fontenc}
\usepackage{amsmath, amsthm, amssymb}
\newcommand\norm[1]{\left\lVert#1\right\rVert}
\newtheorem{theorem}{Theorem}
\newtheorem{lemma}{Lemma}
\usepackage{graphicx}
\author{\IEEEauthorblockN{Osama A. Hanna\IEEEauthorrefmark{2},
        Amr El-Keyi\IEEEauthorrefmark{3} and Mohammed Nafie\IEEEauthorrefmark{1}
        }
    \IEEEauthorblockA{\IEEEauthorrefmark{2}Wireless Intelligent Networks Center (WINC), Nile University, Cairo, Egypt\\
        \IEEEauthorrefmark{3}Department of Systems and Computer Engineering, Carleton University, Ottawa, ON, Canada\\
        \IEEEauthorrefmark{1}EECE Dept., Faculty of Engineering, Cairo University, Giza, Egypt\\
        Email: \{o.ashraf@nu.edu.eg, amr.elkeyi@sce.carleton.ca, mnafie@nileuniversity.edu.eg\}
    }
}
\begin{document}
\title{Degrees of Freedom in Cached MIMO Relay Networks With Multiple Base Stations}
\maketitle
\begin{abstract}
The ability of physical layer relay caching to increase the degrees of freedom (DoF) of a single cell was recently illustrated. In this paper, we extend this result to the case of multiple cells in which a caching relay is shared among multiple non-cooperative base stations (BSs). In particular, we show that a large DoF gain can be achieved by exploiting the benefits of having a shared relay that cooperates with the BSs. We first propose a cache-assisted relaying protocol that improves the cooperation opportunity between the BSs and the relay. Next, we consider the cache content placement problem that aims to design the cache content at the relay such that the DoF gain is maximized. We propose an optimal algorithm
and a near-optimal low-complexity algorithm for the cache content placement problem. Simulation results show significant improvement in the DoF gain using the proposed relay-caching protocol.
\end{abstract}
\section{Introduction}
 Relays are widely used  to extend the coverage of wireless networks without expensive payload backhaul. However, it was shown that relays can not provide degrees of freedom (DoF) gain in conventional wireless networks \cite{cadambe2009degrees}. Recently, it has been discovered that if relays are allowed to deliver cached information beside instantaneous information, the DoF of the network increases \cite{han2015degrees}. The benefits of caching stem from the fact that a currently-delivered data file may be requested again in the future.

 Caching schemes were first proposed for improving the performance of operating systems \cite{belady1966study}. After that, various caching applications have been proposed. Many web caching schemes have been presented such as \cite{borst2010distributed} with the aim of sustaining the data flow of the internet. In \cite{hefeeda2008traffic}, caching was used to reduce the number of hops from the source to the destination, and hence, the delay was reduced. In \cite{barish2000world}, the aim of caching was to reduce the peak traffic by duplicating some contents in distributed memories across the network. This duplication occurs over the off-peak hours. The duplicated contents can be used during the peak hours to reduce network congestion. Small cell  caching was used to improve the area spectral efficiency of video transmission in cellular systems \cite{molisch2011wireless}. An information theoretic formulation of the caching problem was presented in \cite{maddah2014fundamental}. However, all these works separated caching from the physical layer (PHY).

 In \cite{naderializadeh2016fundamental}, it was shown that the one-shot sum-DoF grows linearly with the aggregate cache size in the network. Also, it was recently shown in \cite{han2015degrees} that PHY caching can provide DoF gain in a single cell with a caching relay, as the topology of the relay channel is transformed into a broadcast channel when the requested content exists in the base station (BS) and the relay station (RS) simultaneously. However, the results in \cite{han2015degrees} can not be directly extended to the case of multiple cells.  For example, consider a simple network with $K$ single antenna users, two non-cooperative $M$-antenna BSs and two non-cooperative $M$-antenna RSs each with cache of size $B_C$ (one for each BS). Assuming that the BSs do not have any contents in common, the two BSs can not use the channel simultaneously, and hence,  we have to assign the channel for only one of them at each time instant. Thus, we have no DoF gain over \cite{han2015degrees} in this case. However, in this paper we show that using a single shared RS with $2M$ antennas and cache of size $2B_c$ instead of two separate RSs can greatly improve the DoF beyond those achieved by direct extension of the scheme in \cite{han2015degrees}.

\begin{figure}[t!]
  \centering
    \includegraphics[width=0.45\textwidth]{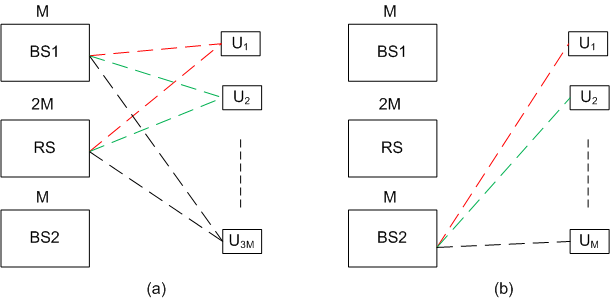}
	  \caption{Cache-induced opportunistic topology transformation. The dashed lines in (a) represents the transmission of cached parity bits where DoF$=3M$ with the help of the shared RS, while the dashed lines in (b) represents the transmission of parity bits not in the RS cache where DoF$=M$.}
\label{fig:1}
\end{figure}

In this paper, we propose a novel scheme to exploit \textit{cache-induced opportunistic MIMO cooperation} \cite{liu2014cache} in relay networks with multiple non-cooperative BSs. In particular, we exploit the benefit of having a shared relay station (RS) which significantly increases the DoF gain over the results of \cite{han2015degrees}. This DoF improvement can be attributed to the ability of a BS to use the shared RS when the other BS cannot benefit from the existence of the RS as illustrated in Fig.~\ref{fig:1}. We propose a cache-assisted relaying protocol using maximum distance separable (MDS) coding which improves the opportunity of MIMO cooperation between the BSs and the RS. Next, we propose a global-optimal cache content placement solution and a low-complexity self-learning algorithm to find a near-optimal cache content placement solution that does not require prior knowledge of the popularity of the content files.

The mathematical notation used in this paper is as follows. The superscripts $(.)^T$ and $(.)^\dagger$ denote the transpose and the Hermitian transpose respectively. The indicator function $1(\mathcal{A})$ returns $1$ when the event $\mathcal{A}$ is true, otherwise it returns $0$.

\section{Cache-assisted Relaying Protocol} \label{sys}
\subsection{System Model}
Consider a cached relay network composed of two non-cooperative \mbox{$M$}-antenna BSs (designated as BS1 and BS2) and a \mbox{$2M$}-antenna shared relay with cache of size \mbox{$2B_c$}. Each BS is connected to a gateway. There are \mbox{$K$} single antenna users requesting files from any of the content gateways via the relay network where \mbox{$K\geq3M$}. We denote the total number of available content files by $L$. We assume that the RS stores a fraction of each content file. Each user has the freedom to request any of the available $L$ files and the BS which connected to the gateway containing the requested file serves the user subsequently possibly with the assistance of the relay.

We assume that each user has a long sequence of requested files. Time is slotted with time slots indexed by $t$ where the duration of each time slot is \mbox{$\tau$}. The index of the file requested by user \mbox{$k$} is denoted by \mbox{$\pi_k$} whose size is given by \mbox{$F_{\pi_k}$}. We define the user request profile (URP) as the $K\times 1$ vector whose $k$th entry contains the index of the file requested by user $k$, i.e., \mbox{$\boldsymbol\pi = [\pi_1,\pi_2,...,\pi_K]^T$} . We assume that \mbox{$\boldsymbol\pi$} has a general distribution which changes over a slow timescale. We also assume that the packets of the content files are delivered to each BS via a high speed backhaul connection and that the RS has no fixed-line backhaul connection to any of the BSs.

Let $\mathcal{N}=\{B_1, B_2, R , D_1, \ldots, D_K \}$ denote the set containing the network nodes where $B_i$, $R$ and $D_k$ denote respectively the $i$th BS, RS, and the $k$th user.  We use $H_{X,Y}$ to denote the matrix containing the coefficients of the channel from transmitter node $Y$ to receiver node $X$ where $X,Y \in \mathcal{N}$. The channel matrices have i.i.d. complex Gaussian entries with zero-mean and unit-variance and are independent of each other. We consider a block fading channel where all the channel matrices remain constant within a time slot but are i.i.d. across the time slots. The transmission power of each BS is bounded by $P_S$ while the  transmission power of the RS is bounded by $P_R$ where  $P_R \leq P_{S}$.
\subsection{Proposed Relay Caching Protocol} \label{protocol}
In this Subsection, we present the proposed cache-assisted relaying protocol. Each BS divides each of its available content files into segments of size \mbox{$L_S$} bits each. These segments of the $l$th file are encoded into \mbox{$L_S + q_l L_S$} parity bits using an MDS rateless code, where \mbox{$q_l\in [0,1]$} is the \textit{cache content placement variable}. The \mbox{$\{q_l L_S\}_{l=1}^L$} parity bits will be cached at the RS. An MDS rateless code generates a codeword of arbitrary length from \mbox{$L_S$} information bits, such that if the decoder receives any \mbox{$L_S$} parity bits, it can recover the original \mbox{$L_S$} information bits. The MDS rateless code can be implemented using Raptor codes \cite{shokrollahi2006raptor} with a small redundancy overhead. As discussed in \cite{han2015degrees}, MDS encoding highly improves the probability of MIMO cooperative transmission. After that, the cache-assisted relaying protocol follows the following phases
\subsubsection{Cache Initialization Phase}
The \mbox{$q_lL_S$} parity bits for every segment of the \mbox{$l$}th content file for $l=1, \ldots ,L$ are transmitted offline to the RS and stored in the RS cache.  We define the \textit{cache content placement vector} as the $L\times 1$ vector whose $l$th component is the cache content placement variable of the $l$th file, i.e., {$\mathbf{q}=[q_1,...,q_L]^T$}. The cache content does not change for each instantaneous realization of \mbox{$\boldsymbol\pi$}. It should be designed according to the distribution of \mbox{$\boldsymbol\pi$} rather than its instantaneous realization. Thus, the change of the RS cache contents occurs over a slow timescale with small average network loading.
\subsubsection{Transmission Phase} \label{step3}
It was previously shown in \cite{han2015degrees} that PHY caching can opportunistically transform the unfavorable relay channel into a more favorable broadcast channel for cached relay networks with a single cell. We will extend this result to cached relay networks with multiple BSs. Moreover, we will investigate the benefit of having a shared RS rather than two separate RSs. This benefit is because each BS can use the shared RS at the time slots where the other BS does not need it.

Let \mbox{$\mathbb{K}_{i}$} denote the set of users that requested files stored in the content gateway of the $i$th BS. Let also \mbox{$q^{(i)}_k=q_{\pi_{O(k,\mathbb{K}_i)}}$} denote the cache content placement variable of the file requested by the $k$th ordered user that belongs to \mbox{$\mathbb{K}_{i}$},  where $i=1, 2$ and \mbox{$O(k,\mathbb{S})$} denotes the $k$th minimum element in the set \mbox{$\mathbb{S}$}. We define the cache content placement vector for the files requested by the users that belong to the set \mbox{$\mathbb{K}_{i}$} as \mbox{$\mathbf{q}_i=[q^{(i)}_1,...,q^{(i)}_{|\mathbb{K}_i|}]^T$}.

\textit{Example 1}: Let us assume that the number of content files is equal to $L=20$ files where files with indices $l=1, \ldots, 10$ are stored in BS1 while files with indices $l=11, \ldots, 20$ are stored in BS2. Let us assume that there are $K=12$ users and that the URP is given by $\boldsymbol\pi=[2,17,19,7,9,3,5,2,20,1,7,6]^T$, i.e., the first user requests the second file while the twelfth user requests file number 6. Hence, $\mathbb{K}_1=\{1,4,5,6,7,8,10,11,12\}$, $\mathbb{K}_2=\{2,3,9\}$, $\mathbf{q}_1=[q_1,q_4,q_5,q_6,q_7,q_8,q_{10},q_{11},q_{12}]^T$ and $\mathbf{q}_2=[q_2,q_3,q_9]^T$.

We also define \mbox{$S^{(k)}\in\{0,1\}$} as the \textit{parity usage state} of user \mbox{$k$}, where \mbox{$S^{(k)}=1$} means that user $k$ has not received all the cached parity bits for the currently requested segment and \mbox{$S^{(k)}=0$} means that user $k$ has received all the cached parity bits for the currently requested segment. We denote the set of users that can be cooperatively served by the $i$th BS and the RS by \mbox{$\mathcal{U}_{C_i}=\{k:S^{(k)}=1,\ k\in \mathbb{K}_i\}$}. Let us define the \textit{cache state} of the $i$th BS by
\begin{equation} \label{eqs1}
S_i=\begin{cases}
1 \text{\ \ if}&  |\mathcal{U}_{C_i}|\geq 3M \\
0 &  \text{otherwise}
\end{cases}
\end{equation}
where \mbox{$i\in \{1,2\}$}. We also define the \textit{total cache state S} as
\begin{equation}
S=\begin{cases} \label{eqs3}
1 \text{\ \ if}&  S_1=1\ \text{or}\ S_2=1 \\
0 &  \text{otherwise}.
\end{cases}
\end{equation}
When \mbox{$S=1$}, the $i$th BS whose $S_i=1$ can cooperate with the RS to transmit \mbox{$3M$} data streams to \mbox{$3M$} different users that belong to the set  \mbox{$\mathcal{U}_{C_i}$} achieving a large DoF gain as shown in Fig.~\ref{fig:1}(a). Hence, there are two transmission modes based on the total cache state $S$\\
 
\textbf{1)BS-only transmission mode (\mbox{$S=0$}):} Since the two BSs can not use the channel simultaneously, the channel is assigned for one BS at each time slot. Since we assume that each user has a long sequence of requested files, i.e., $|\mathbb{K}_1|+|\mathbb{K}_2|=K \geq 3M$, we have two cases:\\
  1) \mbox{$|\mathbb{K}_i|<M$} and \mbox{$|\mathbb{K}_{3-i}|\geq M$}, \mbox{$i\in\{1,2\}$}: In this case we assign the channel to the BS which has \mbox{$|\mathbb{K}_b|\geq M$}, \mbox{$b\in\{1,2\}$} while the other BS will transmit nothing, where \mbox{$b$} denotes the index of the chosen BS to use the channel.\\
  2) Both \mbox{$|\mathbb{K}_{1}|\geq M\ \text{and}\ |\mathbb{K}_{2}|\geq M$}: In this case we choose a BS at random with probability \mbox{$\frac{1}{2}$} to use the channel.\\
  
  The chosen BS randomly selects a subset of \mbox{$M$} users (denoted by \mbox{$\mathcal{U}_b^N$}) from the set \mbox{$\mathbb{K}_b$} and transmits some parity bits that are not in the RS cache to the users in \mbox{$\mathcal{U}_b^N$} using zero-forcing (ZF) beamforming as illustrated in Fig.~\ref{fig:1}(b). In this case, the received signal at user \mbox{$k\in \mathcal{U}_b^N$} is
\begin{equation}
y_{D_k}=H_{D_k,B_b} v_{b_k}s_k+n_{D_k}
\end{equation}
 where \mbox{$n_{D_k} \sim \mathcal{CN}(0,1)$} is the normalized additive white Gaussian noise (AWGN), \mbox{$s_k$} is the data symbol intended for user \mbox{$k$} and \mbox{$v_{b_k}\in \mathbb{C}^M$} is the ZF beamforming vector satisfies
\begin{equation}
H_{D_{k^{'}},B_b} v_{b_k}=0,\ \forall k^{'}\in \mathcal{U}_b^N-\{k\}.
\end{equation}
The data rate of user \mbox{$k$} in this case (\mbox{$S=0$}) is given by
\begin{equation}
r_k^N(v_{b_k})=\min\left(1(k\in \mathcal{U}_b^N)\log_2(1+|H_{D_k,B_b} v_{b_k}|^2),\frac{R_k}{\tau}\right)
\end{equation}
where $R_k$ is the number of non-cached bits that are not delivered to user $k$ for the currently requested segment.
\par   \textbf{2)BS-RS MIMO cooperative transmission mode (\mbox{$S=1$}):} We also have two cases:\\
  1) \mbox{$|\mathcal{U}_{C_i}|<3M$} and \mbox{$|\mathcal{U}_{C_{3-i}}|\geq 3M$}, \mbox{$i\in \{1,2\}$}: In this case we choose the BS which has \mbox{$|\mathcal{U}_{C_b}|\geq 3M$} to use the channel.\\
  2) Both \mbox{$|\mathcal{U}_{C_{1}}|\geq 3M\ \text{and}\ |\mathcal{U}_{C_{2}}|\geq 3M$}: In this case we choose a BS at random with probability \mbox{$\frac{1}{2}$} to use the channel. 
  
  The chosen BS randomly selects a subset of \mbox{$3M$} users (denoted by \mbox{$\mathcal{U}_b^{C}$}) from the set \mbox{$\mathcal{U}_{C_b}$}. The BS and the RS cooperatively transmit some cached parity bits to the users in \mbox{$\mathcal{U}_b^{C}$} using ZF beamforming as illustrated in Fig.~\ref{fig:1}(a). Let $h_k$ denote the $3M\times 1$ augmented channel vector containing the coefficeints of the channels from the BS and RS to the $k$th user, i.e.,   that \mbox{$h_k=\left[H_{D_k,B_b}^T, H_{D_k,R}^T\right]^T$}. In this case, the received signal at user \mbox{$k\in \mathcal{U}_b^{C}$} can be written as
\begin{equation}
y_{D_k}^{'}=h_k v_{b_k}^{'}s_k^{'}+n_{D_k}
\end{equation}
 where \mbox{$s_k^{'}$} is the cached data symbol intended for user \mbox{$k$} and \mbox{$v_{b_k}^{'}\in \mathbb{C}^{3M}$} is the ZF beamforming vector satisfies
\begin{equation}
h_{k^{'}} v_{b_k}^{'}=0,\ \forall k^{'}\in \mathcal{U}_b^{C}-\{k\}.
\end{equation}
 The data rate of user \mbox{$k$} in this case (\mbox{$S=1$}) is given by
\begin{equation}
r_k^{C}(v_{b_k}^{'})=\min\left(1(k\in \mathcal{U}_b^{C})\log_2(1+|h_k v_{b_k}^{'}|^2),\frac{R^{'}_k}{\tau}\right)
\end{equation}
where $R^{'}_k$ is the number of cached bits that are not delivered to user $k$ for the currently requested segment.
\par \textit{Example 2}: Assume the same setup of example 1 and suppose that $\mathcal{U}_{C_1}=\mathbb{K}_1$ and $\mathcal{U}_{C_2}=\mathbb{K}_2$. Then, $S_1=1$, $S_2=0$, because $|\mathcal{U}_{C_1}|>3M$, $|\mathcal{U}_{C_2}|<3M$, and hence, $S=1$. The transmission will occur according to case 1 in the BS-RS MIMO cooperative transmission mode.

At the beginning of each time slot, each BS determines \mbox{$\mathbb{K}_{i}$} and \mbox{$\mathcal{U}_{C_i}$}, $i=1,2$ based on the cache content placement vector \mbox{$\mathbf{q}$} and the user requests. After that, each BS determines \mbox{$S_i$} based on (\ref{eqs1}) then sends it to the RS. The RS determines the total cache state \mbox{$S$} based on (\ref{eqs3}) and selects a BS to use the channel based on the discussion we made earliers. Then, the RS sends the local CSI to the selected BS which in turn forms the global CSI \mbox{$H_{D_k,B_b}$} and \mbox{$H_{D_k,R}$}. The selected BS computes all the control variables (\mbox{$\mathcal{U}_b^C / \mathcal{U}_b^{N}\ \text{and }\{v_{b_k}\} / \{v_{b_k}^{'}\}$}) and sends them to the RS. Then, the transmission begins based on the total cache state \mbox{$S$}.\\

\subsubsection{Segment Decoding Phase}
Each user stores the received parity bits in the reassembling buffer until its size reaches \mbox{$L_S$}. The received segment is then MDS decoded and the buffer is cleared to store the next segment.

\section{DoF-Optimal cache content placement} \label{opt}
\subsection{DoF in Cached Relay Network}
In traditional relay networks, the DoF is defined as the ratio of the sum rate and \mbox{$\log$} SNR as SNR tends to \mbox{$\infty$}. In \cite{han2015improving}, the DoF definition of a cached relay network was introduced as
\begin{equation} \label{DoFdef}
\text{DoF}(\mathbf{q})=\lim_{P_{S}+P_R\to \infty}\mathbb{E}\left\lbrace \frac{R_\Sigma (P_{S}, P_R, B_C,\{F_l\})}{\log_2 \left( P_{S}+P_R\right) } \right\rbrace
\end{equation}
where the expectation is taken over the URP \mbox{$\boldsymbol\pi$}, cache state \mbox{$S$} and the channel matrices. \mbox{$R_\Sigma (P_{S}, P_R, B_C,\{F_l\})$} is the sum rate received by the \mbox{$K$} users. To have a non-zero DoF, we have to scale \mbox{$B_C \text{ and }F_l$} with the SNR. Thus, we assume
\begin{equation}
B_C= \tilde{B}_C \log_2 (P_S+P_R)
\end{equation}
\begin{equation}
F_l= \tilde{F}_l \log_2 (P_S+P_R)\ \forall l\in\{1,...,L\}.
\end{equation}
We also assume that $L_s= \tilde{L}_s \log_2 (P_S+P_R)$. The following theorem characterizes the DoF of our cached relay network.
\begin{theorem}
For given \mbox{$\mathbf{q}$}, \mbox{$2\tilde{B}_C\geq \sum_{l=1}^Lq_l\tilde{F}_l$} and as $\tilde{L}_S\to \infty$, the DoF of the cached relay network with two BSs is given by
\begin{equation} \label{DoF}
\text{DoF}(\mathbf{q})=M\left( 1+2\mathbb{E}_{\boldsymbol\pi}\left\lbrace P\{S=1\}\right\rbrace \right)
\end{equation}
where $P\{S=1\}$ is the fraction of time at which BS-RS MIMO cooperative transmission is used for a given URP $\boldsymbol\pi$.
\end{theorem}
\begin{proof}
Refer to Appendix \ref{proof1}.
\end{proof}
\subsection{BS-RS MIMO Cooperative Transmission Probability}
From theorem 1, we see that the overall performance of the cached relay network depends strongly on the probability of BS-RS MIMO cooperative transmission mode \mbox{$P\{S=1\}$}. It was previously shown in \cite{han2015degrees} that MDS-coded caching  significantly improves the probability of BS-RS MIMO cooperative transmission mode over naive random caching. The following theorem gives an expression for the BS-RS MIMO cooperative transmission probability.
\begin{theorem}
For given $\mathbf{q}$, URP $\boldsymbol\pi$ and as \mbox{$\tilde{L}_S\to\infty$}, the fraction of time at which BS-RS MIMO cooperative transmission is used can be given by:
\begin{equation}\label{coop_prob}
P\{S=1\}=\frac{m_{\boldsymbol\pi}^{(1)}+m_{\boldsymbol\pi}^{(2)}-2m_{\boldsymbol\pi}^{(1)}m_{\boldsymbol\pi}^{(2)}}{f_1(\boldsymbol\pi)+f_2(\boldsymbol\pi)+4m_{\boldsymbol\pi}^{(1)}m_{\boldsymbol\pi}^{(2)}}
\end{equation}
where $f_i(\boldsymbol\pi)=3-3I_i-(5-3I_{3-i})m_{\boldsymbol\pi}^{(i)}$, $I_i=1(|\mathbb{K}_i|<M)$, \mbox{$m_{\boldsymbol\pi}^{(i)}=\frac{3M}{|\mathbb{K}_i|}\min_k\{\tilde{q}^{(i)}_k\}$}, $i\in\{1,2\}$ and \mbox{$\{\tilde{q}^{(i)}_k\}$} can be obtained from \mbox{$\{q^{(i)}_k\}$} by Algorithm \ref{alg1}.
\end{theorem}
\begin{proof}
Refer to Appendix \ref{proof2}.
\end{proof}
\par We define \mbox{$\tilde{\mathbf{q}}_i=[\tilde{q}^{(i)}_1,...,\tilde{q}^{(i)}_{3M}]$} as the \textit{modified cache content placement vector} for the $i$th BS. We can get \mbox{$\tilde{\mathbf{q}_i}$} from \mbox{$\mathbf{q}_i\ \forall i\in \{1,2\}$} by the following algorithm
\begin{algorithm}\caption{Getting \mbox{$\tilde{\mathbf{q}_i}$} from \mbox{$\mathbf{q}_i\ \forall i\in \{1,2\}$}} \label{alg1}
\textbf{Step 0: If} \mbox{$|\mathbb{K}_i|<3M$}, put all the entries of \mbox{$\mathbf{q}_i$} in \mbox{$\tilde{\mathbf{q}_i}$}, and extend the size of $\tilde{\mathbf{q}_i}$ to be \mbox{$3M$} by adding zeros at the end of it, then terminate the algorithm.\\
\textbf{Else}, put the largest \mbox{$3M$} entries of \mbox{$\mathbf{q}_i$} in \mbox{$\tilde{\mathbf{q}_i}$} with any order. Then, replace the largest \mbox{$3M$} entries of \mbox{$\mathbf{q}_i$} by zeros.\\
\textbf{Step 1:} Get the largest entry of \mbox{$\mathbf{q}_i$}, and add it to the smallest entry in \mbox{$\tilde{\mathbf{q}_i}$}, i.e., $\tilde{q}^{(m)}_i=\tilde{q}^{(m)}_i+\max_k \{q^{(k)}_i\}$, where $m$ is the index of the smallest entry in $\tilde{\mathbf{q}_i}$. Then, replace the largest entry in \mbox{$\mathbf{q}_i$} by zero and repeat \textbf{step 1} until all the entries of \mbox{$\mathbf{q}_i$} become zeros.
\end{algorithm}
\par \textit{Example 3}: Let us consider the same setup of example 2. Let us assume that the vectors $\mathbf{q}_1$ and $\mathbf{q}_2$ are given by $\mathbf{q}_1=[0.3,0.1,0.4,0.05,0.6,0.3,0.4,0.1,0.2]^T$, and the vector $\mathbf{q}_2=[0.1,0.2,0.4]^T$. Then, $\tilde{\mathbf{q}}_1=[0.4,0.4,0.6,0.35,0.4,0.3]^T,\tilde{\mathbf{q}}_2=[0.1,0.2,0.4,0,0,0]^T$.
\par It is observed that if $ |\mathbb{K}_2|<M $ and $|\mathbb{K}_1|=3M$, then $I_2=1, \min_k\{\tilde{q}^{(2)}_k\}=0$, and hence, (\ref{coop_prob}) can be reduced to the cooperative transmission probability in the case of single BS which was given in \cite{han2015degrees}.
\par For clarity, we have considered a case with two BSs and a simple antenna configurations to highlight the main ideas. However, the results and algorithms can be easily extended for arbitrary number of BSs with more general antenna settings.
\subsection{DoF Optimization}
In this Subsection we formulate the DoF optimization problem. We aim to find the cache content placement vector \mbox{$\mathbf{q}$} based on the content popularity statistics that maximizes the DoF. From (\ref{DoF}), to maximize the DoF, we need to maximize \mbox{$\mathbb{E}_{\boldsymbol\pi}\left\lbrace P\{S=1\}\right\rbrace$}. We denote \mbox{$D=\{\mathbf{q}:q_l\in[0,1]\ \forall l\in\{1,...,L\},\ \sum_{l=1}^L\tilde{F}_lq_l\leq 2\tilde{B}_C \}$} as the feasible set of cache content placement vectors. The DoF optimization problem is formulated as
\begin{equation}
\mathcal{P}:\ \max_{\mathbf{q}\in D} p(\mathbf{q})=\mathbb{E}_{\boldsymbol\pi}\left\lbrace \frac{m_{\boldsymbol\pi}^{(1)}+m_{\boldsymbol\pi}^{(2)}-2m_{\boldsymbol\pi}^{(1)}m_{\boldsymbol\pi}^{(2)}}{f_1(\boldsymbol\pi)+f_2(\boldsymbol\pi)+4m_{\boldsymbol\pi}^{(1)}m_{\boldsymbol\pi}^{(2)}} \right\rbrace.
\end{equation}
\mbox{$\mathcal{P}$} is a stochastic optimization problem. Moreover, the objective function is neither concave nor convex for unknown URP distribution. We propose a global-optimal solution for $\mathcal{P}$ which requires the knowledge of the distribution of $\boldsymbol\pi$ and a low-complexity algorithm to find a near-optimal solution which does not require explicit knowledge of the URP distribution.
\subsection{Global-Optimal Solution}
For a given distribution of \mbox{$\boldsymbol\pi$}, the objective function of problem $\mathcal{P}$ is a convex function. Problem $\mathcal{P}$ involves the maximization of a convex function over a feasible set given by a polyhedron. As a result, the optimal solution is one of the vertices of the polyhedron. Let us assume that \mbox{$\boldsymbol\pi$} can take the value \mbox{$\boldsymbol\pi^{(l)}$} with probability $f_l$,  $l=1,...,L^K$. We denote \mbox{$\tilde{\mathbf{q}_i}(\boldsymbol\pi^{(l)}), i=1,2$} as the modified cache content placement vectors for the realization \mbox{$\boldsymbol\pi^{(l)}$}. To solve our optimization problem, we transform \mbox{$\mathcal{P}$} into an equivalent deterministic optimization problem \mbox{$\mathcal{P}_E$} by adding auxiliary variables \mbox{$t_1,...,t_{2L^K}$} and non-negative slack variables \mbox{$s_1,...,s_{1+L}$}. We assume that \mbox{$w=[q_1,...,q_L,t_1,...,t_{2L^K},s_1,...,s_{1+L}]^T$}. Thus we have
\begin{equation} \label{P_E}
\mathcal{P}_E:\ \max_{w\in D_E}p_E(w)=\sum_{i=1}^{L^K}f_i\frac{m_{\boldsymbol\pi^{(i)}}^{(1)}+m_{\boldsymbol\pi^{(i)}}^{(2)}-2m_{\boldsymbol\pi^{(i)}}^{(1)}m_{\boldsymbol\pi^{(i)}}^{(2)}}{f_1(\boldsymbol\pi^{(i)})+f_2(\boldsymbol\pi^{(i)})+4m_{\boldsymbol\pi^{(i)}}^{(1)}m_{\boldsymbol\pi^{(i)}}^{(2)}}
\end{equation}
where
\begin{equation} \label{D_E}
\begin{aligned} D_E=&\{\mathbf{q},\{t_i\},\{s_i\}:\sum_{l=1}^L\tilde{F}_lq_l+s_1=2\tilde{B}_C,q_l+s_{1+l}=1,\\
&q_l\geq 0,s_1\geq 0,s_{1+l}\geq 0,t_i=\frac{3M}{|\mathbb{K}_1|}\min_k\{\tilde{q}^{(1)}_k(\boldsymbol\pi^{(i)})\},\\
&t_{L^K+i}=\frac{3M}{|\mathbb{K}_2|}\min_k\{\tilde{q}^{(2)}_k(\boldsymbol\pi^{(i)})\}\ \forall l\in \{1,...L\},\\
& i\in \{1,...L^K\}, k\in \{1,...K\} \}\end{aligned}
\end{equation}
is the feasible set of \mbox{$w$}. \mbox{$D_E$} is a polyhedron and \mbox{$\mathcal{P}_E$} is a convex maximization problem with linear constrains. The global-optimal solution of \mbox{$\mathcal{P}_E$} has the following property
\begin{theorem}
Given $K$ and the distribution of $\boldsymbol\pi$, the global-optimal solution of \mbox{$\mathcal{P}_E$} is a vertex of the polyhedron \mbox{$D_E$}.
\end{theorem}
\begin{proof}
Refer to Appendix \ref{proof3}.
\end{proof}
From theorem 3, we can solve \mbox{$\mathcal{P}_E$} by enumerating the vertices of \mbox{$D_E$}. Majthay and Whinston suggested a finite algorithm for solving this problem \cite{majthay1974quasi}. Unfortunately, the complexity of the Majthay-Whinston algorithm grows exponentially with the number of files $L$. Moreover, since the size of the support of \mbox{$\boldsymbol\pi$} grows exponentially with the number of users \mbox{$K$}, the complexity of the Majthay-Whinston algorithm also grows exponentially with \mbox{$K$}. Also, it is not easy in practice to obtain the distribution of \mbox{$\boldsymbol\pi$}. To overcome these challenges, we propose a low-complexity algorithm that does not require explicit knowledge of the distribution of \mbox{$\boldsymbol\pi$}.
\subsection{Low-Complexity Near-Optimal Solution}
Sampling-based algorithms have been previously used to solve stochastic optimization problems \cite{ruszczynski2003stochastic}. We follow on this direction to solve the problem using \mbox{$N$} independent samples of \mbox{$\boldsymbol\pi$} which denoted by \mbox{$\hat{\boldsymbol\pi}^{(1)},...,\hat{\boldsymbol\pi}^{(N)}$}. Then, the sample average approximation problem of \mbox{$\mathcal{P}$} is
\begin{equation} \label{optpro}
\max_{\mathbf{q}\in D}\ p_N(\mathbf{q})=\frac{1}{N}\sum_{i=1}^N \frac{m_{\hat{\boldsymbol\pi}^{(i)}}^{(1)}+m_{\hat{\boldsymbol\pi}^{(i)}}^{(2)}-2m_{\hat{\boldsymbol\pi}^{(i)}}^{(1)}m_{\hat{\boldsymbol\pi}^{(i)}}^{(2)}}{f_1(\hat{\boldsymbol\pi}^{(i)})+f_2(\hat{\boldsymbol\pi}^{(i)})+4m_{\hat{\boldsymbol\pi}^{(i)}}^{(1)}m_{\hat{\boldsymbol\pi}^{(i)}}^{(2)}}.
\end{equation}
\mbox{$D$} is a polyhedron and (\ref{optpro}) is a convex maximization problem with linear constrains. We propose a low-complexity algorithm (Algorithm \ref{alg2}) to solve this problem. The idea of Algorithm \ref{alg2} is inspired from the simplex method which can be founded in \cite{murty1983linear}.
\begin{algorithm}\caption{Low-complexity cache content placement algorithm without knowledge of URP distribution} \label{alg2}
\textbf{Step 0:} Obtain \mbox{$N$} realizations of URP \mbox{$\{\hat{\boldsymbol\pi}^{(1)},...\hat{\boldsymbol\pi}^{(N)}\}$}. Let \mbox{$\mathbf{q}^0=[0,...,0]^T$} be the initial vertex of the polyhedron \mbox{$D$}. Set the step index \mbox{$k=0$}.\\
\textbf{Step 1:} Starting at \mbox{$\mathbf{q}^k$}, if there is no adjacent vertex which has a greater value for \mbox{$p_N(\mathbf{q})$}, terminate the algorithm. Otherwise, assign the adjacent vertex where \mbox{$p_N(\mathbf{q})$} is maximum to \mbox{$\mathbf{q}^{k+1}$}. Put \mbox{$k=k+1$} and repeat \textbf{Step 1}.
\end{algorithm}
In step 1 of Algorithm \ref{alg2}, we utilize the subroutine in \cite{murty1983linear} to find the adjacent vertices of a vertex \mbox{$\mathbf{q}^k$}. The simplex method is efficient, as well as Algorithm~\ref{alg2}.
\section{Numerical Results} \label{num}
In this Section, we numerically evaluate the performance of our proposed scheme and compare it with various baselines. We begin with a simple simulation to compare the results of Algorithm \ref{alg2} with the global-optimal solution results. A cached relay network with two BSs is considered with $L=60$ content files ($30$ files for each BS) and $K=12$ users. We assume that the sizes of all the content files are equal. Each BS has $M=2$ antennas and the RS has $4$ antennas. We assume that the power constrains of the BSs and the RS are equal, i.e., $P_S=P_R=P$. Each user requests a content file independent on the other users. The content popularity of each content gateway is modeled by the Zipf distribution \cite{breslau1999web} with popularity skewness parameter $\gamma=2$. We assume that a user requests a file from the content gateway of BS1 or BS2 with probability $\frac{1}{2}$.
\begin{figure}[!ht]
  \centering
    \includegraphics[width=0.48\textwidth]{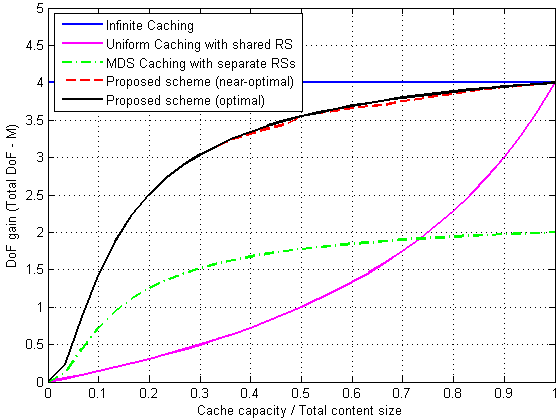}
	  \caption{DoF gain versus cache capacity for $L=60$, $K=12$ and $M=2$}
 \label{fig:3}
\end{figure}
\begin{figure}[!ht]
  \centering
    \includegraphics[width=0.48\textwidth]{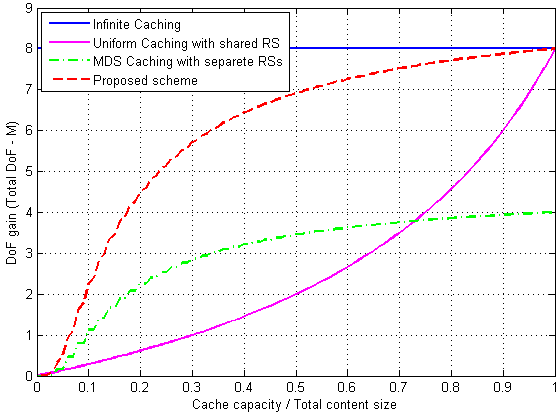}
	  \caption{DoF gain versus cache capacity for $L=100$, $K=24$ and $M=4$}
 \label{fig:2}
\end{figure}

\par In Fig.~\ref{fig:3}, we plot the DoF gain versus the cache capacity. The results of our proposed near-optimal solution using $N=400$ URP samples are compared with the proposed global-optimal solution results and with other various schemes. The first scheme is the \textit{infinite caching} scheme where the RS can store all the content files. The \textit{uniform caching} scheme where $q_l=\frac{2\tilde{B}_C}{L\tilde{F}}\ \forall l\in\{1,...,L\}$ is the second one. The third is the \textit{MDS-coded caching with separate RSs} scheme where each BS has a separate RS as in \cite{han2015degrees}. From Fig.~\ref{fig:3}, we can see that the results of the low-complexity algorithm are very close to the global solution results for a moderate value of $N$.
\par Now, we change our setup to be $L=100$, $K=24$ and $M=4$. As $L$ and $K$ increase, the complexity of the Majthay-Whinston algorithm increases rapidly and the global-optimal solution takes a very long time. Thus, we use Algorithm \ref{alg2} to calculate the optimizing vector $\mathbf{q}$ using $N=10000$ samples.

\par From Fig.~\ref{fig:2}, we can observe the significant improvement in the DoF gain due to having a shared RS instead of two separate RSs. We can also observe the improvement in the DoF gain due to using MDS-coded caching. Also, we can see that the gap between the results of \cite{han2015degrees} and the proposed scheme increases as the cache capacity increases. These results demonstrate that our proposed solution has significant performance gain over the case where there are two separate RSs and also over the other caching schemes. This improvement occurs even in the case where the cache capacity is small compared to the total content size.
\section{Conclusion} \label{con}
In this paper, we have proposed a cache-assisted relaying protocol for cached relay networks with multiple BSs. No cooperation between the BSs is needed. We have exploited the benefit of having a shared RS to provide large DoF gain. An optimal and a low-complexity near-optimal algorithms for the cache content placement problem have been proposed with the objective of maximizing the DoF gain. Numerical results have been provided to show the performance improvement of the proposed scheme over various baseline schemes.
\appendix
\subsection{DoF Formulation (Proof of Theorem 1)} \label{proof1}
The average instantaneous sum rate received by the $K$ users can be given by
\begin{equation}
\begin{aligned}
\mathbb{E}_{H}\Big\{ \max_{v_{b_k},v_{b_k}^{'}} R_\Sigma (P_{S}&, P_R, B_C,\{F_l\}) \Big\} \\
=\sum_{j=1}^2 \max_{v_{j_k},v_{j_k}^{'}}\Biggl( &P\{S=0,b=j\}\sum_{k=1}^K\mathbb{E}_{H}\{r_k^N(v_{j_k})\} \\
& +P\{S=1,b=j\}\sum_{k=1}^K\mathbb{E}_{H}\{r_k^C(v_{j_k}^{'})\} \Biggr)
\end{aligned}
\end{equation}
where $b\in \{1,2\}$ is the index of the chosen BS to use the channel. The beamforming vectors $v_{b_k}$ and $v_{b_k}^{'}$ satisfy
\begin{equation}
\begin{aligned}
P\{&S=0,b=j\}\sum_{k\in\mathcal{U}_b^N}\mathbb{E}_{H}\norm{v_{j_k}}_2^2 \\
& +P\{S=1,b=j\}\sum_{k\in\mathcal{U}_b^C}\mathbb{E}_{H}\norm{v_{j_k}^{'}\begin{bmatrix}\mathbf{I}_M &\mathbf{0}\\ \mathbf{0} &\mathbf{0}\end{bmatrix}}_2^2 \leq P_S,\\
P\{&S=1,b=j\}\sum_{k\in\mathcal{U}_b^C}\mathbb{E}_{H}\norm{v_{j_k}^{'}\begin{bmatrix}\mathbf{0} &\mathbf{0}\\ \mathbf{0} &\mathbf{I}_{2M}\end{bmatrix}}_2^2\leq P_R,\\
&\hspace{5cm} \forall\ j\in\{1,2\}
\end{aligned}
\end{equation}
where $\mathbf{I}_M$ is the $M\mathsf{x}M$ identity matrix. From the definition of the DoF of cached relay networks in (\ref{DoFdef}), we have
\begin{equation}
\begin{aligned}
\text{DoF}(\mathbf{q})&=\lim_{P_{S}+P_R\to \infty}\mathbb{E}\left\lbrace \frac{R_\Sigma (P_{S}, P_R, B_C,\{F_l\})}{\log_2 \left( P_{S}+P_R\right) } \right\rbrace\\
&=\mathbb{E}_{\boldsymbol\pi} \Big\{ M.(1-P\{S=1\})+3M.P\{S=1\}\Big\}\\
&=M\left( 1+2\mathbb{E}_{\boldsymbol\pi}\left\lbrace P\{S=1\}\right\rbrace \right).
\end{aligned}
\end{equation}
\subsection{MIMO Cooperation Probability (Proof of Theorem 2)} \label{proof2}
Let us define \mbox{$r^C_{i_k}=\mathbb{E}_{\boldsymbol\pi}\{r_{k}^C(v^{'}_{i_k})|S_i=1, k \in \mathcal{U}_i^{C}\}$} and
\mbox{$r^N_{i_k}=\mathbb{E}_{\boldsymbol\pi}\{r_{k}^N(v_{i_k})|S_i=0, k \in \mathcal{U}_i^N\}$}, $i=1,2$. \\
We begin with proving the following lemma
\begin{lemma}
\begin{equation}
\lim_{P_{S} \to \infty}\frac{r^C_{1_k}}{r^N_{1_k}} = \lim_{P_{S} \to \infty}\frac{r^C_{2_k}}{r^N_{2_k}} = 1.
\end{equation}
\end{lemma}
\begin{proof}
From the power constraint \mbox{$P_{S} \geq P_R$}, it can be easily shown that
\begin{equation} \label{eq1}
 \frac{r^C_{1_k}}{r^N_{1_k}} \leq 1.
\end{equation}
On the other hand, we have 
\begin{equation}
\frac{r^C_{1_k}}{r^N_{1_k}} \geq \frac{\log_2(\gamma_{R}P_R)}{\log_2(\gamma_{S}P_{S})}
\end{equation}
where \mbox{$\gamma_R,\gamma_S$} are finite valued factors which depend on the channel coefficients. Assume that \mbox{$\frac{\gamma_RP_R}{\gamma_SP_{S}}={\gamma}$} and that the ratio \mbox{$\frac{P_{S}}{P_R}$} is finite. Hence,
\begin{equation}
\frac{r^C_{1_k}}{r^N_{1_k}} \geq \frac{\log_2(\gamma)}{\log_2(\gamma_SP_{S})}+1
\end{equation}
and thus we have
\begin{equation} \label{eq2}
\lim_{P_{S} \to \infty}\frac{r^C_{1_k}}{r^N_{1_k}} \geq 1.
\end{equation}
From \eqref{eq1} and \eqref{eq2}, we have
\begin{equation} \label{eq5}
\lim_{P_{S} \to \infty}\frac{r^C_{1_k}}{r^N_{1_k}} = \lim_{P_{S} \to \infty}\frac{r^C_{2_k}}{r^N_{2_k}} = 1.
\end{equation}
\end{proof}
We first assume that BS1 has a priority in the cooperative mode, i.e., BS1 chooses the cooperative mode whenever it is possible regardless of the cache state of BS2. Let \mbox{$N_s$} denote the number of time slots required to transmit the current segment of user $k$ which served by BS1. As \mbox{$\tilde{L}_S\rightarrow \infty$}, we have \mbox{$N_s\rightarrow \infty$} and the average rate of user $k$ can be given by
\begin{equation} \label{eq3'}
r_k=\frac{L_s}{N_s\tau}.
\end{equation}
\par Let \mbox{$P^{(1)}_i=\frac{1}{N_s}\sum_{t=1}^{N_s} 1(S_i(t)=1)$}, $i=1,2$. In this case \mbox{$r_k$} can be expressed as
\begin{equation} \label{eq3}
r_k=P^{(1)}_1\frac{3M}{|\mathbb{K}_1|}r^C_{1_k}+(1-P^{(1)}_1)(1-P^{(1)}_2).\frac{1}{2-I_2}.\frac{M}{|\mathbb{K}_1|}r^N_{1_k}.
\end{equation}
The first term in (\ref{eq3}) is the rate of user $k$ in the case of $S_1=1$, while the second term is the rate in the case of non-cooperative transmission where both $S_1=0$ and $S_2=0$. In the non-cooperative mode, when $|\mathbb{K}_2|<M$ only BS1 can use the channel, while when $|\mathbb{K}_2|\geq M$ BS1 uses the channel with probability $\frac{1}{2}$, and hence, the factor $1/(2-I_2)$ appears.
\par Using MDS-coded caching aligns the transmission of the cached parity bits as much as possible and the total number of transmitted cached parity bits in \mbox{$N_s$} time slots is equal to \mbox{$\min_k\{\tilde{q}^{(1)}_k\}L_s$}. Dividing this number of bits by the average cooperative time $P^{(1)}_1N_s\tau$ yields
\begin{equation}\label{eq4}
r^C_{1_k}=\frac{\min_k\{\tilde{q}^{(1)}_k\}L_s}{P^{(1)}_1N_s\tau}.
\end{equation}
From \eqref{eq3'}, \eqref{eq3} and \eqref{eq4}, we have
\begin{equation}
3P^{(1)}_1+\frac{1}{2-I_2}(1-P^{(1)}_1)(1-P^{(1)}_2)\frac{r^N_{1_k}}{r^C_{1_k}}=\frac{3P^{(1)}_1}{\frac{3M}{|\mathbb{K}_1|}\min_k\{\tilde{q}^{(1)}_k\}}.
\end{equation}
We can take the limit as $P_S\to\infty$ and use lemma 1 to get
\begin{equation} \label{eq6}
3P^{(1)}_1+\frac{1}{2-I_2}(1-P^{(1)}_1)(1-P^{(1)}_2)=\frac{3P^{(1)}_1}{\frac{3M}{|\mathbb{K}_1|}\min_k\{\tilde{q}^{(1)}_k\}}.
\end{equation}
In a similar way if we assume that user k is served by BS2, we will get
\begin{equation} \label{eq7}
3P^{(1)}_2(1-P^{(1)}_1)+\frac{1}{2-I_1}(1-P^{(1)}_1)(1-P^{(1)}_2)=\frac{3P^{(1)}_2(1-P^{(1)}_1)}{\frac{3M}{|\mathbb{K}_2|}\min_k\{\tilde{q}^{(2)}_k\}}.
\end{equation}
From \eqref{eq6} and \eqref{eq7}, we have
\begin{equation}
P^{(1)}_2=\frac{\frac{3M}{|\mathbb{K}_2|}\min_k\{\tilde{q}^{(2)}_k\}}{6-3I_1-(5-3I_1)\frac{3M}{|\mathbb{K}_2|}\min_k\{\tilde{q}^{(2)}_k\}}
\end{equation}
\begin{equation}
P^{(1)}_1=\frac{\frac{3M}{|\mathbb{K}_1|}\min_k\{\tilde{q}^{(1)}_k\}(1-P^{(1)}_2)}{6-3I_2-(5-3I_2+P_2^{(1)})\frac{3M}{|\mathbb{K}_1|}\min_k\{\tilde{q}^{(1)}_k\}}.
\end{equation}
Hence, the probability of the BS-RS MIMO cooperative transmission given that BS1 has a priority in the cooperative mode is given by
\begin{equation}
P^{(1)}\{S=1\}=P^{(1)}_1+P^{(1)}_2-P^{(1)}_1P^{(1)}_2.
\end{equation}
If we assume that BS2 has a priority in the cooperative mode, we will get in a similar way
\begin{equation}
P^{(2)}_1=\frac{\frac{3M}{|\mathbb{K}_1|}\min_k\{\tilde{q}^{(1)}_k\}}{6-3I_2-(5-3I_2)\frac{3M}{|\mathbb{K}_1|}\min_k\{\tilde{q}^{(1)}_k\}}
\end{equation}
\begin{equation}
P^{(2)}_2=\frac{\frac{3M}{|\mathbb{K}_2|}\min_k\{\tilde{q}^{(2)}_k\}(1-P^{(2)}_1)}{6-3I_1-(5-3I_1+P^{(2)}_1)\frac{3M}{|\mathbb{K}_2|}\min_k\{\tilde{q}^{(2)}_k\}}
\end{equation}
and hence, the probability of the BS-RS MIMO cooperative transmission given that BS2 has a priority in the cooperative mode is given by
\begin{equation}
P^{(2)}\{S=1\}=P^{(2)}_1+P^{(2)}_2-P^{(2)}_1P^{(2)}_2
\end{equation}
where \mbox{$P^{(2)}_i=\frac{1}{N_s}\sum_{t=1}^{N_s} 1(S_i(t)=1)$}, $i=1,2$ given that BS2 has a priority in the cooperative mode. To preserve fairness between the two BSs, we take the average performance of the two assumptions. Thus we have
\begin{equation}
P\{S=1\}=\frac{m_{\boldsymbol\pi}^{(1)}+m_{\boldsymbol\pi}^{(2)}-2m_{\boldsymbol\pi}^{(1)}m_{\boldsymbol\pi}^{(2)}}{f_1(\boldsymbol\pi)+f_2(\boldsymbol\pi)+4m_{\boldsymbol\pi}^{(1)}m_{\boldsymbol\pi}^{(2)}}.
\end{equation}
\subsection{Proof of Theorem 3} \label{proof3}
From (\ref{P_E}), it can be shown that $p_E(w)\leq \sum_{i=1}^{L^K}f_i=1$. Thus, $p_E(w)$ is finite over $\mathcal{D}_E$. In the case of $I_1=1$ or $I_2=1$, the expression of $p_E(w)$ in (\ref{P_E}) can be reduced to \cite{han2015degrees} which was shown to be convex over $\mathcal{D}_E$. In the case of $I_1=0$ and $I_2=0$, the Hessian of $p_E(w)$ is a block-diagonal matrix whose diagonal block~$i,i\in \{1,...L^K\}$ is given by (\ref{hess}). For simplicity, we use $m_{i1},m_{i2}$ instead of $m_{\boldsymbol\pi^{(i)}}^{(1)},m_{\boldsymbol\pi^{(i)}}^{(2)}$.
\begin{equation} \label{hess}
D_{i}=\begin{bmatrix} \frac{f_i12(5-4m_{i2})(m_{i2}-1)^2}{(6+4m_{i1}m_{i2}-5m_{i1}-5m_{i2})^3} &\frac{f_i12(1-m_{i1})(1-m_{i2})}{(6+4m_{i1}m_{i2}-5m_{i1}-5m_{i2})^3}\\
\frac{f_i12(1-m_{i1})(1-m_{i2})}{(6+4m_{i1}m_{i2}-5m_{i1}-5m_{i2})^3} &\frac{f_i12(5-4m_{i1})(m_{i1}-1)^2}{(6+4m_{i1}m_{i2}-5m_{i1}-5m_{i2})^3}\end{bmatrix}.
\end{equation}
Also, the determinant of $D_i$ is given as
\begin{equation}
\text{det}(D_i)=\frac{576(1-m_{i1})^2(1-m_{i2})^2}{(6+4m_{i1}m_{i2}-5m_{i1}-5m_{i2})^5}f_i.
\end{equation}
From (\ref{D_E}), we have $m_1 \leq 1,\ m_2\leq 1$, and hence,
\begin{equation}
\begin{aligned}
6+4m_1m_2-5m_1-5m_2 &= 6-(5-4m_2)m1-5m_2\\
&\geq 6-(5-4m_2)-5m_2\\
&=1-m_2 \geq 0.
\end{aligned}
\end{equation}
Thus, $D_i$ is positive-semidefinite. Hence, the hessian of $p_E(w)$ is positive-semidefinite, so $p_E(w)$ is a convex function over $\mathcal{D}_E$. Using \cite{benson1995concave} and property \cite{shokrollahi2006raptor}, $\mathcal{P}_E$ has at least one global-optimal solution which must be a vertex of $\mathcal{D}_E$.

\end{document}